\documentclass[a4paper,fleqn]{cas-sc}

\usepackage{graphicx} 
\usepackage{amsfonts}
\usepackage{amsmath, amsthm}
\usepackage{amssymb}
\usepackage{bbm}
\usepackage{hyperref}
\usepackage{algorithm}
\usepackage[dvipsnames]{xcolor}
\usepackage{comment}
\usepackage{algpseudocode}
\usepackage[normalem]{ulem} 
\usepackage{siunitx}

\newcommand{\E}[0]{\mathbb{E}}
\renewcommand{\P}[0]{\mathbb{P}}

\newtheorem{theorem}{Theorem}

\newtheorem{lemma}{Lemma}
\newtheorem{corollary}{Corollary}
\newtheorem{definition}{Definition}

\newcommand{\eat}[1]{}

\renewcommand{\P}{\mathbf{P}}
\renewcommand{\E}{\mathbf{E}}

\ExplSyntaxOn
\cs_set:Npn \__make_tbl_caption:nn #1#2
{
    \skip_vertical:N \l_tbl_abovecap_skip
    \parbox{\l_tbl_width_dim}
    {
        \raggedright
        \sffamily\small
        \textbf{\color{scolor}#1:}~#2\par
    }
    \skip_vertical:N \l_tbl_belowcap_skip
}
\ExplSyntaxOff

\begin{document}

\title[mode=title] {Splittable Spanning Trees and Balanced Forests\\ in Dense Random Graphs}

\tnotemark[1]

\tnotetext[1]{
An earlier version of this paper appeared in the proceedings of LATIN 2026 under the title
\emph{Splittable Spanning Trees and Balanced Forests in Dense Random Graphs}  \cite{LATIN-Original}.
}
\shortauthors{Gillman et al.}
\shorttitle{Splittable Spanning Trees in Random Graphs}

\author[1]{David Gillman}[
    orcid=0000-0001-7747-8625]
\ead{dgillman@ncf.edu}

\author[2]{Jacob Platnick}[
    orcid=0009-0007-2144-163X]
\cormark[1]

\ead{jplatnick3@gatech.edu}

\author[2,3]{Dana Randall}[
    orcid=0000-0002-1152-2627]
\ead{randall@cc.gatech.edu}

\affiliation[1]{
    organization={New College of Florida},
    city={Sarasota},
    state={Florida},
    country={USA}
}

\affiliation[2]{
    organization={Georgia Institute of Technology},
    city={Atlanta},
    state={Georgia},
    country={USA}
}
\affiliation[3]{
    organization={Santa Fe Institute},
    city={Santa Fe},
    state={New Mexico},
    country={USA}
}

\cortext[cor1]{Corresponding author}









\thispagestyle{empty}

\begin{abstract}
We consider the probability that a spanning tree chosen uniformly at random from a graph can be
partitioned into $k$ trees of equal size with the removal of $k-1$ edges, for $k$ fixed. Such {\it splittable}
spanning trees allow us to sample {\it balanced forests}, or forests whose components are of equal
size, with applications in redistricting, network algorithms, and image decomposition. Cannon et
al. showed that spanning trees on grid and grid-like graphs on $N$ vertices are splittable 
with probability at least $N^{-2k}$, leading to the first rigorous sampling algorithm
for balanced forests for any class of graphs. Focusing on the complementary case of dense
random graphs, we show that random spanning trees also have inverse polynomial probability
of being splittable (with probability at least
$N^{-k/2}$ for both the $G(N,p)$ and $G(N,m)$ models when 
$p=\Omega(1/\log N))$, giving the first dense
class of graphs where partitions of equal size can be sampled efficiently. In addition, we present
an infinite family of graphs satisfying several properties that have been conjectured to ensure splittability (e.g., 
Hamiltonicity) and prove that the probability random spanning trees are
splittable is exponentially small. As a consequence, 
a family of widely-used Markov chain algorithms for sampling equitable partitions will fail on
this family of graphs if their state spaces are restricted to partitions of equal or nearly equal size,
suggesting barriers for sampling balanced partitions in some classes of sparse graphs.

\end{abstract}

\maketitle

\section{Introduction}

Sampling balanced forests of a graph 
has garnered a lot of interest lately, especially in the context of political redistricting.
Given an integer $k>1$ and a graph $G=(V,E)$ with $|V|=kn$, the problem is to sample uniformly from the set of forests 
of $G$ consisting of $k$ trees, each of size $n$.
Sampling balanced forests is equivalent to sampling partitions of the vertices of $G$ into $k$ connected components $V=(V_1, \dots, V_k)$ with probability proportional to $\Pi_{i=1}^k w(V_i)$, where the weight $w(V_i)$ is the number of spanning trees of $V_i$. This is known as the {\it spanning tree weight}.

When used for redistricting,  random weighted equitable partitions are sampled to form an ensemble against which proposed maps can be tested for political gerrymandering.
Each vertex of the graph represents a small indivisible region, or precinct, and the components of a partition represent the districts from which representatives are elected.
This methodology has been used for court cases and various amicus briefs \cite{DukeReCom1,
zhao2022}.  While the requirements for acceptable political maps vary from state to state, mathematicians have generally agreed that the spanning tree weight is a desirable baseline because spanning trees are easy to generate and they typically produce maps with compact districts \cite{procacciaTuckerFoltz2021compactness}; moreover, when other weightings are preferable, parallel tempering has been successfully used to interpolate between the spanning tree weight and other useful metrics \cite{DukeReCom2,DukeReCom1}.
Uniformly generating balanced forests is also desirable in the contexts of clustering in biological and social networks \cite{fortunato2010,junker} and image decomposition~\cite{camilus2012}.  

\subsection{Algorithms for sampling balanced forests}
There are three main families of algorithms for sampling balanced forests, all of which rely on a random spanning tree's ``splittability'' into $k$ pieces, which we now define.   
\begin{definition}
    A tree is $k$-splittable if it contains $k-1$ edges that, when removed, create a graph whose $k$ connected components all have the same number of vertices.
\end{definition}

The most widely used algorithm for sampling balanced forests is {\it ReCom}  (for recombination) \cite{recom} and its variants, which provably converge to the distribution defined by spanning tree weight \cite{DukeReCom1,revrecom}\footnote{For simplicity of exposition, we refer to this group of algorithms simply as ``ReCom,'' although our results apply to any of the modified versions which have provably correct stationary distributions \cite{DukeReCom1,revrecom}.}. 
ReCom uses Markov chains that walk on balanced forests (or weighted partitions) by iteratively choosing two adjacent pieces, finding a spanning tree on the vertices in the union of these pieces, and then partitioning this joint region into two equal-sized pieces by splitting the spanning tree in two, if possible.  When a tree cannot be equitably split by removing a single edge, the move is rejected and the Markov chain proceeds to propose a new recombination move.

For this approach to be efficient and reliable, it is necessary that the recombination move connects the state space of partitions, the chain is rapidly mixing, and subgraphs of $G$ contain splittable spanning trees with at least an inverse polynomial probability. Unfortunately, it has been shown that there are potential issues with each of these three requirements.  In particular, no graphs except the complete graph are known to be connected under perfectly balanced ReCom moves. Moreover, Charikar et al. gave a family of graphs where strictly following ReCom moves will be slowly mixing
\cite{charikar2022complexity}. 
Natural modifications to ReCom have been proposed to increase its connectivity and potentially its mixing rate by allowing some slack on the piece sizes. Yet, even with this allowance, the state space has only been shown to be connected in special cases~\cite{akitaya2022reconfiguration_recombination,cannontri,charikar2022complexity}. 

Another family of sampling algorithms works by generating uniform spanning trees on all of $G$ (using, for example, the loop-erased random walk (LERW) algorithm of Wilson \cite{wilson}) and checking whether each tree is $k$-splittable, which can be done in $O(N)$ time by a depth-first search. Note that any splittable tree has a unique set of $k-1$ edges whose removal leaves a balanced forest (see Lemma~\ref{lem:unique}), so rejection sampling will generate such forests uniformly.  If the tree is not splittable, it is resampled.  This algorithm produces samples from exactly the right conditional distribution, but the running time depends on the rejection rate.  

Markov chains that walk on the set of (not necessarily balanced) forests with $k$ components also exist,  such as the up-down walk \cite{charikar2022complexity}. 
This chain has an $O(N^2\log(N))$ mixing time on all graphs and is known to be ergodic, but again we need to employ rejection sampling whenever the forests are imbalanced.
The remaining balanced forests will be generated almost uniformly and all other forests will be rejected and resampled, so the overall algorithm will be fast only if the rejection rate is not too large. 
Note that the probability that a random tree is splittable is polynomially related to the likelihood that a random forest is balanced (see Lemma~\ref{lem:equiv}),
so either both the LERW and up-down walk approaches to sampling balanced forests are efficient or neither is.

As all of these algorithms rely on the probability that random spanning trees are splittable, it is worthwhile to characterize which classes of graphs have this property.
Charikar et al.~\cite{charikar2022complexity} conjectured that many graphs are splittable with inverse polynomial probability, including rectangular regions of $\mathbb{Z}^2$.
Recently, Cannon, Pegden and Tucker-Foltz~\cite{cannon2024grid} gave a stunning proof of this conjecture for planar grid graphs and grid-like graphs.  
Planarity is critical to their argument because they rely on the dual of the loop-erased random walk to bound the likelihood of specific families of forests being generated to show that these already occur with at least $1/N^2$ probability, suppressing dependencies on~$k$.

Little is known about sufficient conditions to ensure splittability in other classes of graphs. Akitaya et al. \cite{akitaya2022reconfiguration_recombination} have conjectured that Hamiltonicity is sufficient, and Cannon \cite{cannontri} has conjectured that it is sufficient for a graph to be an induced subgraph of the triangular lattice. Similarly, little is known about sufficient conditions to ensure fast mixing of ReCom. It has been conjectured that it is necessary to allow slack on the piece sizes to guarantee connectivity of the state space \cite{akitaya2022reconfiguration_recombination,cannontri}, or to avoid slow mixing \cite{charikar2022complexity}. 

\subsection{Obstacles to splittability}
While it is tempting to believe the conventional wisdom that there are simple conditions on graphs that are sufficient to ensure their splittability with high probability, such as Hamiltonicity, sparsity, and planarity (e.g.,~\cite{akitaya2022reconfiguration_recombination,cannontri,charikar2022complexity}), we give counterexamples to each of these conjectures.  In particular, in Theorem~\ref{thm:slack} (Section~\ref{sec:split}), we present an infinite family of graphs on $3n$ vertices that are Hamiltonian induced subgraphs on the triangular lattice, where the only equitable partitions of the graph into three connected components are the~$n$ cyclic partitions of the Hamiltonian cycle at $i, \ i+n, \ i+2n,$ for $0\leq i \leq n-1$.  However, the probability of a random spanning tree being splittable on this family is exponentially small.  This implies that all of the algorithms for generating balanced partitions mentioned above, including ReCom, LERWs, and the up-down walk, provably require exponential time.

This infinite family of graphs on $3n$ vertices also constitutes a troubling counterexample for ReCom with slack.  In particular, if the state space for ReCom requires the subgraphs to be perfectly balanced, then the state space is disconnected; however, if even the smallest amount of slack is permitted and the state space includes partitions into pieces of size $\{n-1, n, n+1\}$, then ReCom will require exponential time to find a balanced partition.  This example gives strong evidence that the hope to characterize which sparse graphs are splittable will be more challenging than previously believed.

\subsection{Splittability of dense random graphs}
We prove that dense graphs, on the other hand, do enjoy high splittability probability.  We show that if a graph is sufficiently dense, it will be splittable with high  enough probability that we can efficiently generate balanced partitions using either the LERW or the up-down chain.
First, we observe that the rejection overhead of ReCom on $K_{kn}$ is proportional to the inverse of the probability of trees being splittable. 
To see this, we can consider a ``bin-mixing'' Markov chain on $k$ bins of $n$ balls each, where a move consists of combining two bins and randomly mixing the balls between them such that each bin ends up with $n$ balls. We see that when running ReCom on the complete graph, each move where it finds a $2$-splittable spanning tree is equivalent to a move on the bin-mixing Markov chain, and each one where it does not is a move where it does nothing. ReCom on this graph can be considered as a lazy version of the bin-mixing Markov chain. 
If the bin-mixing Markov chain mixes in time $t_{\text{mix}}$, then ReCom mixes in time proportional to $t_{\text{mix}}$ divided by the probability that a random spanning tree of $K_{2n}$ is $2$-splittable.

More precisely, we show in Theorem~\ref{thm:complete} that for any constant $k$, the probability that a random spanning tree in the complete graph $K_{kn}$ is $k$-splittable is 
$$\P[\text{A uniform spanning tree on\ } K_{kn} \text{\ is splittable}] = \ \Theta(n^{-(k-1)/2}).$$

\noindent Furthermore, in Theorem~\ref{thm:dense} we prove our main theorem, showing that a polynomial fraction of the set of $k$-component forests on dense random graphs will be balanced, for any constant $k$.  Specifically,  we show that if $N=nk$ and $G$ is a random connected graph chosen from $G(N,p)$ with constant $p$, then the probability that a random spanning tree is $k$-splittable is at least 
$$\P[\text{A uniform spanning tree on\ } G \in G(N,p) \text{\ is splittable}] = \ \Omega(n^{-(k-1)/2}).$$  
The 
same holds for $G(N,m)$ for random graphs of comparable expected density, and similar polynomial results hold for $p=\Omega(1/\log(n))$.
This complements the work of Cannon et al.~\cite{cannon2024grid}, who established a similar result for the $m \times n$ grid and a related class of “grid-like” graphs.
Note that Cannon et al.~\cite{cannon2024grid} showed the proportion of spanning trees that are splittable on the $\sqrt{N} \times \sqrt{N}$ grid scales by at least $N^{2-2k}$, so these new bounds on splittability for dense graphs are much more favorable.

\color{black}

These are the first results showing any dense graphs are splittable with sufficiently large probability to use rejection sampling, which is a requirement for the simple algorithm, the up-down walk, and ReCom.
Interestingly, the results on $G(N,m)$ are more straightforward to prove. We then derive the corresponding results for $G(N,p)$ from those, rather than the other way around.

One of the challenges to bounding splittability on dense graphs 
is that the spanning trees that appear with higher probability than others are actually {\it less likely} to be splittable.
While all spanning trees are equally likely to appear in a random graph, some trees are more likely to appear in graphs that contain fewer spanning trees,  making those trees more likely to be selected.  For example, a star with three leaves is more likely than a path of three edges to be selected as a random spanning tree of a random graph on four vertices. In fact, Alon et al.~\cite{alon20} have shown that the diameter of random spanning trees of dense graphs grows as $\sqrt{N}$, as was previously shown by Szekeres \cite{SOURCE-UST} for the complete graph.  Trees with smaller diameter {would seem} less likely to be splittable than those with large diameter, so inverse polynomial bounds on the splittability of uniform spanning trees on dense graphs may be more surprising than similar results on sparse graphs such as the grid, where the diameter of a uniform spanning tree is linear.

\section{Preliminaries}

We begin by formalizing some of the important lemmas relating probabilities of weighted partitions, balanced forests, and splittability as well as some of the algorithms that are commonly studied.

\subsection{The spanning tree weight and splittability}
The spanning tree weight is often used to generate equitable partitions of a graph because the conditional distribution on balanced forests is uniform.
\begin{definition}
    The spanning tree weight of a partition $P$ of a graph $G$ is the product of the number of spanning trees in each piece.
\end{definition}
\noindent We will denote it as $\tau(P)$. Similarly, we let $\tau(G)$ be the number of spanning trees in a graph $G$, as it is the weight of the partition of $G$ into a single piece. We also define the number of $k$-splittable spanning trees in~$G$ as $\tau_k(G)$.

\begin{lemma}\label{lem:unique}
    Given any tree $T(V,E)$ with $kn$ vertices, for some integer $k$, there exists at most one set $F\subseteq E$ of $k-1$ edges such that the removal of $F$ from $T$ creates $k$ connected components, each with size $n$.\label{uniqueness_lemma}
   
\end{lemma}
\begin{proof}
    When $k=1$, clearly $F$ can only be the empty set. Assume this lemma holds for splits into at most $k-1$ pieces and proceed by induction. 
    
    If there exist two sets $F_1, F_2$ satisfying the conditions of the Lemma, consider the graph created by contracting the components of $T$ that would result from removing~$F_1$. This is a tree with $k$ vertices, and therefore it has a leaf, corresponding to a subtree $S$ of $T$ with $n$ vertices and exactly one edge $e$ between it and $V-S$. $F_2$ cannot contain any edge in $S$, as if it did, it would split~$S$ into multiple connected components in $T-F_2$, only one of which can contain an endpoint of $e$, and the others of which would be a proper subset of $S$. A proper subset of $S$ has fewer than $n$ vertices, which is impossible.
    
    It follows that $S$ is contained in a connected component of $T-F_2$. Since $|S| = n$, $S$ must be equal to that component and $F_2$ must contain $e$. 
    By induction, both $F_1 - \{e\}$ and $F_2 - \{e\}$ split $T-S$ into $k-1$ pieces. 
    Therefore, $F_1 - \{e\} = F_2 - \{e\}$ and $F_1 = F_2$.
\end{proof}

It will be convenient to talk about both balanced forests and splittable trees.  The following establishes their close relationship.
   A similar result appears in \cite{cannon2024grid}.

\begin{lemma}\label{lem:equiv} 
    For any family $\mathcal{G}$ of connected graphs and constant $k$, the following are equivalent:\begin{enumerate} \item There exists a polynomial $p_k(\cdot)$ such that, for any $G=(V,E)\in \mathcal{G}$, a random spanning tree of $G$ is $k$-splittable with probability at least ${1}/{p_k(|V|)}.$
    \item There exists a polynomial $q_k(\cdot)$ such that, for any $G=(V,E)\in \mathcal{G}$, a random $k$-partition of $G$ weighted by the spanning tree weight is balanced with probability at least ${1}/{q_k(|V|)}$.
    \end{enumerate}   
\end{lemma}\begin{proof}
    We note that the probability that a random spanning tree is $k$-splittable is the ratio of the number of $k$-splittable spanning trees to the number of spanning trees of a graph, and the probability that a random $k$-partition weighted by the spanning tree weight is balanced is the ratio of the total weight of all balanced $k$-partitions to the total weight of all $k$-partitions. Then as a direct consequence of the next two lemmas, the ratio of the probability that a random spanning tree of $G$ is $k$-splittable to the probability that a random 
    $k$-partition of $G$ weighted by the spanning tree weight is balanced is between $N^{-2(k-1)}$ and $N^{3(k-1)}$. 
\end{proof}

\begin{lemma}
    For any connected graph $G$ with $N$ vertices, the ratio between the total weight of all $k$-partitions of $G$ and the number of spanning trees of $G$ is between $N^{-2(k-1)}$ and $N^{k-1}$. \label{first_equiv_lemma}
\end{lemma}
\begin{proof} 
The weight of a partition $P$ under the spanning tree weight is the product of the number of spanning trees in each piece. Then this weight is equal to the number of spanning forests whose connected components are the pieces of $P$. 

Consider the space of spanning forests in $G$ with~$k$ connected components.
We see that each such forest can be made into a spanning tree by adding $k-1$ edges between its components, and since every spanning tree can be split into $k$ pieces by arbitrarily removing $k-1$ edges, every spanning tree can be made this way. Since there are at most $N^2$ edges in $G$, the number of spanning trees is at most $N^{2(k-1)}$ times the number of such forests. Similarly, since every such forest can be made by removing $k-1$ edges from a spanning tree, and each tree contains $\binom{N-1}{k-1}$ forests with~$k$ connected components, the number of such forests is at most $\binom{N-1}{k-1}<N^{k-1}$ times the number of spanning trees.  \end{proof}
\begin{lemma}
    For any connected graph $G=(V,E)$ with $N$ vertices, the ratio of the total weight of all balanced $k$-partitions of $G$ and the number of $k$-splittable spanning trees of $G$ is between $N^{-2(k-1)}$ and $1$.
\end{lemma}\begin{proof}
By Lemma~\ref{uniqueness_lemma}, each $k$-splittable spanning tree uniquely defines a $k$-balanced spanning forest, so the number of them is at most the number of spanning trees. 
Further, each $k$-splittable  tree can be obtained from a balanced forest by adding $k-1$ edges, so there are at most $\binom{|E|}{k-1}<N^{2(k-1)}$ times as many spanning trees as balanced forests. \end{proof}

\subsection{Algorithms for generating balanced forests}
It will be useful to formalize the three main algorithms for sampling balanced partitions.
First, the simplest algorithm is splitting a full spanning tree (see Algorithm~\ref{alg:simple}). This algorithm will run a single stage and either return a balanced partition sampled uniformly according to the spanning tree weight, or fail. When using it in practice, the algorithm will be run until it succeeds in finding a balanced $k$-partition. It is described in \cite{cannon2024grid}. 
The algorithm first generates a random spanning tree. 
Then it generates a partition by splitting the tree into $k$ equal pieces, if that is possible. 
Then it accepts the partition with probability equal to the inverse of the number of spanning trees that would generate that partition. 
The algorithm finds $k-1$ edges in the spanning tree to split it equitably as follows. 
If $k-1$ such edges exist, they are unique, according to Lemma~\ref{uniqueness_lemma}.
Each edge, when removed, splits the tree into two pieces whose sizes are multiples of $n$,
and this property uniquely characterizes these $k-1$ edges.

\begin{algorithm}[!t]
    \caption{Splitting a Spanning Tree}\label{alg:simple}
\begin{algorithmic}[1]
\Statex    {\bf \underbar{Input:} A connected graph $G$ with $kn$ vertices.}
 \State   Sample a uniform random spanning tree $T$ of $G$ using Wilson's algorithm.
 \State   Initialize the set $K$ of edges to remove from the tree: $K \leftarrow \emptyset$. 
 \State   Run a depth-first search on $T$, traversing each edge in turn. 
 For each edge $e$, determine the sizes of the two pieces that would be created by removing $e$ from the tree. If {the sizes are multiples of $n$} add $e$ to $K$.
    \If{$|K| = k-1$}
    \State remove the edges in $K$ from the tree to create a forest;
    \State let $P$ be the partition created by its connected components. 
 \Else  
 \State fail. \EndIf
\State    Let $H$ be the multigraph created by contracting all vertices in each piece in $P$.
\State    Use Kirchhoff's theorem to compute $\tau(H)$.
\State    With probability $1/\tau(H)$, \Return $P$. Otherwise, fail.
\end{algorithmic}
\end{algorithm}

Next, we will examine two Markov-chain-based algorithms. 
The up-down forest walk, described in \cite{charikar2022complexity},  moves among forests with $N-k$ edges (see Algorithm~\ref{alg:updown}). 
The number of times~$t$ that the loop is iterated depends on the mixing time, which was shown to be  $O(N^2\log(N))$ in the worst case~\cite{charikar2022complexity}. \color{black} 

\begin{algorithm}[!t]
    \caption{Up-Down Forest Walk}\label{alg:updown}
\begin{algorithmic}[1]
\Statex   {\bf \underbar{Input:} a connected graph $G$ with $kn$ vertices and a forest $F$ on $G$ with $kn-k$ edges.} 
\Statex   {\bf Begin with any forest with $k$ components}
\State{\bf repeat:}
\State   Select an edge $e$ uniformly at random from $G-F$ 
\State Add  $e$ to $F$.
\State Remove a random edge from the cycle in $F$.
\State \bf{end repeat}
\State\Return $F$.
\end{algorithmic}
\end{algorithm}

The final algorithm we describe is Reversible ReCom as described in \cite{revrecom} (see Algorithm~\ref{alg:recom}). This is also a Markov chain algorithm that only moves between balanced partitions, so it always returns a valid balanced partition; however, its mixing time is unknown and sometimes exponential. 

\begin{algorithm}[!t]
    \caption{Reversible ReCom (without Slack)}\label{alg:recom}
\begin{algorithmic}[1]
\Statex {\bf \underbar{Input:} A connected graph $G$ with $kn$ vertices and a balanced $k$-partition $P=\{P_1,\dots,P_k\}$.}
\Statex  {\bf Begin with any partition of $G$ into $k$ components of size $n$}
\State{\bf repeat:}
\State Select pieces $P_i, P_j$ uniformly at random.
 \If{$P_i$ and $P_j$ are not adjacent} return $P$. 
 \Else{} 
\State Sample a uniform random spanning tree $T$ of $P_i\cup P_j$ using Wilson's algorithm.
\State  Run a depth-first search to see if $T$ has an edge $e$ that splits $T$ into two pieces of size~$n$.
 \If{such an edge $e$ exists} label the two pieces that are created by removing~$e$ as $P_i', P_j'$
\Else{} \Return~$P$.
\EndIf  
\EndIf
\State Let $E(P_i',P_j')$ be the number of edges spanning $P_i',P_j'$.
\State With probability $1/E(P_i',P_j')$, remove 
 $P_i, P_j$ from $P$ and add $P_i',P_j'$.
 \State {\bf end repeat}
\State \Return $P$.
\end{algorithmic}
\end{algorithm}

All of the listed Markov chain algorithms are known to converge to a stationary distribution that selects each forest with equal probability, including both reversible ReCom and Forest ReCom \cite{DukeReCom1,revrecom,charikar2022complexity}. However, ReCom Markov chains may not be irreducible \cite{akitaya2022reconfiguration_recombination}, so may have multiple  stationary distributions depending on the starting partition. Moreover, guarantees on the mixing times of ReCom are unknown, with some cases known to require exponential time \cite{charikar2022complexity}.

\section{Obstacles, even with slack}\label{sec:split}
We now present the family of graphs for which the probability of a random spanning tree being splittable  is exponentially small, and which highlight potential problems with ReCom when allowing slack.

\begin{definition}
    A tree is $k$-splittable {\em with slack
    $j$} if there exist $k-1$ edges that, when removed, create a graph whose $k$ connected components have between $n-j$ and $n+j$ vertices, inclusive, for some value of $n$. 
    Let $\tau^j(G)$ be the number of splittable spanning trees of $G$ with slack $j$. 
\end{definition}
\begin{theorem} 
    There exists an infinite family $\{G_n: n\geq 6, n~\mathrm{even}\}$ of graphs with $3n$ vertices such that each graph $G_n$ \
    \begin{enumerate}
        \item is a Hamiltonian induced subgraph of the triangular lattice,
        \item contains exactly $n$ partitions into three pieces of size $n$, namely the equitable partitions of the 
        Hamiltonian cycle, 
        \item is such that 
        $\tau^0(G_n) / \tau^1(G_n)$ 
        {is bounded above by an exponential $\exp({-\Omega(n)})$.}
    \end{enumerate}
    \label{thm:slack}
\end{theorem}

\begin{corollary}\label{cor}
    There exists an infinite family of graphs on $3n$ vertices, $n \geq 6$, for which ReCom without slack has a disconnected state space and ReCom with slack cannot sample partitions of equal size in polynomial time.
\end{corollary}

First, it will be useful to introduce some structures called {\it ladders}.

\begin{definition}
    A ladder in the triangular lattice is the induced subgraph created by two parallel lines whose endpoints  each neighbor one of the other line's endpoints.
\end{definition}

\begin{lemma}
If $L_1$ and $L_2$ are ladders in the triangular lattice consisting of $n_1$ and $n_2$ vertices, respectively, with $n_1>n_2>1$, then $$(3^{n_1-n_2})\; \tau(L_2)\ \geq\ \tau(L_1)\ \geq \ (2^{n_1-n_2})\; \tau(L_2).$$\label{lemma:ladder}
\end{lemma}

\begin{proof}
    It is clear that any two ladders with the same number of vertices must be isomorphic graphs. Then consider the ladder $L'$ with $n_1-1$ vertices that is a subset of $L_1$ with one of the end vertices $v$ removed.  That vertex $v$ is adjacent to two other vertices: one on its line and one that is an endpoint of the other line, forming two edges incident to $v$. Then any spanning tree of $L'$ combined with either one of those two edges creates a unique spanning tree of $L_1$, so $\tau(L_1)\geq 2\;\tau(L')$. Then repeating by induction until $L'$ has the same number of vertices as $L_2$, we get $\tau(L_1)\geq (2^{n_1-n_2})\; \tau(L_2)$.
    
To show $(3^{n_1-n_2})\; \tau(L_2) \geq  \tau(L_1)$, we again consider adding a single vertex $v$ to a ladder $L_2$ to form a longer ladder~$L'$. Every spanning tree of $L'$ has either  $1$ or $2$ edges incident to $v$. If it has $1$ edge incident to $v$, each such tree can be constructed with one of the two edges to $v$ and a spanning tree of $L_2$, so there are $2\tau(L_2)$ of them. If instead the spanning tree of $L'$ includes both edges incident to $v$, where $u_1$ and $u_2$ are the vertices in $L_2$ that are the other endpoints of these edges, then the spanning tree cannot contain $(u_1,u_2)$. Removing $(u_1,v)$ and $(u_2,v)$ and adding $(u_1,u_2)$ creates a spanning tree of~$L_2$. Then each such spanning tree of~$L'$ can be constructed by removing $(u_1,u_2)$ from a spanning tree of $L_2$ and adding $(u_1,v)$ and $(u_2,v)$, so there are at most $\tau(L_2)$ of them, so by summing these cases and using induction we get $3^{n_1-n_2}\tau(L_2) \geq \tau(L_1).$
\end{proof}

We can now construct the family of graphs discussed in Theorem~\ref{thm:slack}.

\newenvironment{proofof}[1]{%
 \par\addvspace{\medskipamount}
  \par\noindent\textit{Proof of #1.}
}{%
  \hfill\par
}
\begin{figure}[pos=ht]
    \centering 
    \includegraphics[width=0.24\linewidth, angle=270]{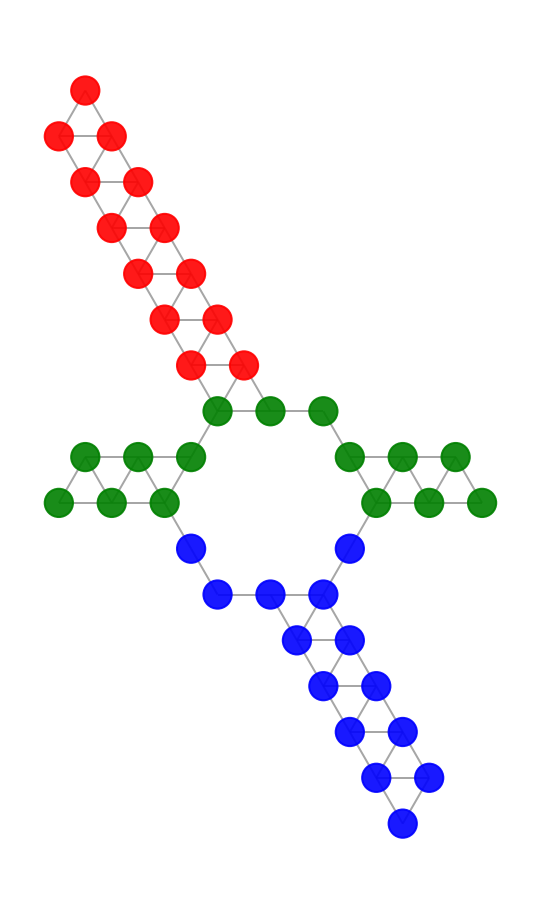}\hspace{.3in} 
    \includegraphics[width=0.24\linewidth, angle=270]{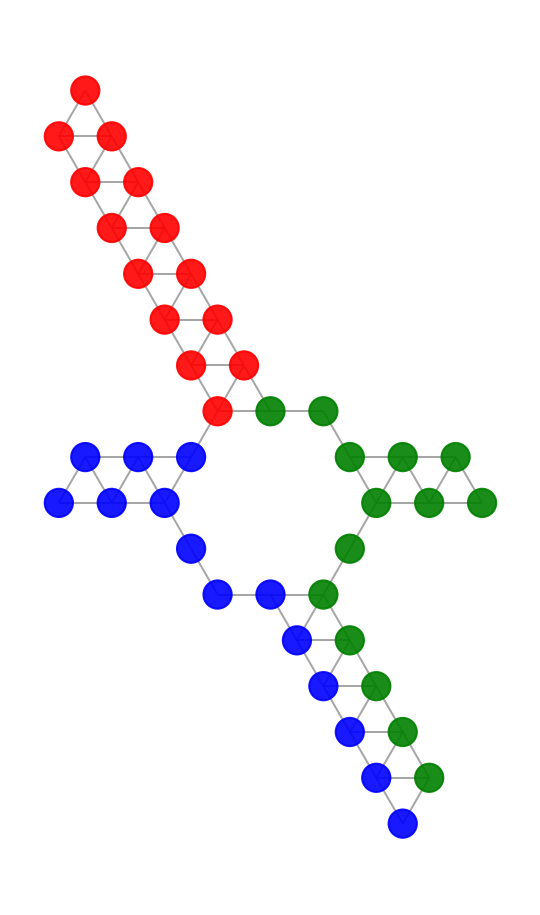}
    \caption{Graph described in Theorem~\ref{thm:slack} with $n=14$:
   (a) An almost balanced partition, with piece sizes $\{13, 14, 15\} $   and (b) an example of a balanced partition, with piece sizes $\{14, 14, 14\}$.}
    \vspace{-1.2in} (a)\hspace{2.8in} (b)\hspace{2.5in}
    \vspace{1in}
    \label{fig:slack-graph}
\end{figure}

\begin{proofof}{Theorem~\ref{thm:slack}}
We will construct the family of graphs ${G_n}$ depicted in Figure~\ref{fig:slack-graph}. For any $n$, $G_n$ consists of a hexagon with $2$ edges on each side, attached to $4$ ladders. For any even $n$, $n \geq 6$, we can expand the ladders such that $G_n$ contains $3n$ vertices, with the red piece having a ladder with $n-1$ vertices; the green piece containing~$7$ vertices of the hexagon and each of its two ladders containing an additional $(n-6)/2$ vertices, for a total of $n+1$ vertices; and the blue piece containing $5$ vertices of the hexagon and an additional $n-5$ vertices on its ladder. This is always a valid almost-balanced partition into pieces of size $\{n-1, n, n+1\}$. 
This graph contains a Hamiltonian cycle that traverses the hexagon clockwise, taking a clockwise detour out and back along each ladder.

To show property (2) of the theorem, we claim that the only way to split this graph into three connected pieces of size $n$ is to split the Hamiltonian cycle into three paths of size $n$. This will follow from the fact that each ladder has fewer than $n$ vertices. 
Consider any way to partition $G_n$ into three connected pieces, $A, B, C,$ of size $n$. If one of the pieces, $A$, is not a path along the Hamiltonian cycle, then it consists of at least two disconnected paths along the Hamiltonian cycle. Traversing the Hamiltonian cycle clockwise, between the last vertex of one of these paths and the first vertex of the next path, there is a nonempty set of vertices $S$ in $B \cup C$. $S$ cannot contain a vertex of the hexagon, because in that case $A$ would be disconnected. But $S$ also cannot be a subset of one of the ladders, because in that case $|S| < n$ and $S$ would be disconnected from the rest of $B \cup C$. This is a contradiction.

We now show property (3) of the theorem.
Given a partition, say a ladder is ``monochromatic'' if it is contained in one piece of the partition.
Consider the almost balanced partition shown in Figure~\ref{fig:slack-graph}(a). There are four maximal monochromatic ladders that contain $3n-6$ vertices in total. 

Now consider any balanced partition $P$ of $G_n$ into induced subgraphs $P_1, P_2, P_3$. Consider the subgraph~$L_i$ of $P_i$ induced by the maximal ladders in $P_i$. 
By the discussion above, $P_i$ consists of components of $L_i$ connected in a chain by paths. Therefore, $\tau(P_i)$ = $\tau(L_i)$. 
By Lemma~\ref{lemma:ladder} it will be enough to show that the monochromatic ladders of $P$ contain at most $3n - \Omega(n)$ vertices total.

There are four ladders and only three pieces, so either there are fewer than four monochromatic ladders, in which case we are done, or one piece $P_i$ of the partition contains two ladders. Those two ladders must be adjacent in the clockwise traversal of the Hamiltonian cycle. 
Any pair of adjacent ladders in $G_n$ contain $n + \Omega(n)$ total vertices, of which $P_i$ must be missing $\Omega(n)$.
\end{proofof}

\begin{proofof}{Corollary~\ref{cor}}
    When using ReCom (or its variants) without slack on the family of graphs constructed for the proof of Theorem \ref{thm:slack}, Property (2) implies that no move is possible, as any piece's position along the Hamiltonian cycle determines the positions of the other two pieces. Hence, any move splitting and recombining two of them cannot modify the partition, and ReCom fails to make any nontrivial moves.
    As the graph can only be partitioned along the Hamiltonian cycle, there are only $n$ balanced partitions of the graph up to relabeling. {We also see that when allowing any slack, we must allow the partitions with the smallest amount of slack: those with piece sizes $n-1,n,n+1$. Then any sample made from ReCom with slack is exponentially more likely to be unbalanced than balanced, as the example high-weight unbalanced partition described in Theorem~\ref{thm:slack} is clearly reachable as a ReCom move from a balanced partition along the Hamiltonian cycle.}
\end{proofof}

\section{Splittability of Dense  Graphs}\label{sec:dense}

\begin{figure}[pos=ht]
    \includegraphics[width=0.45\linewidth]{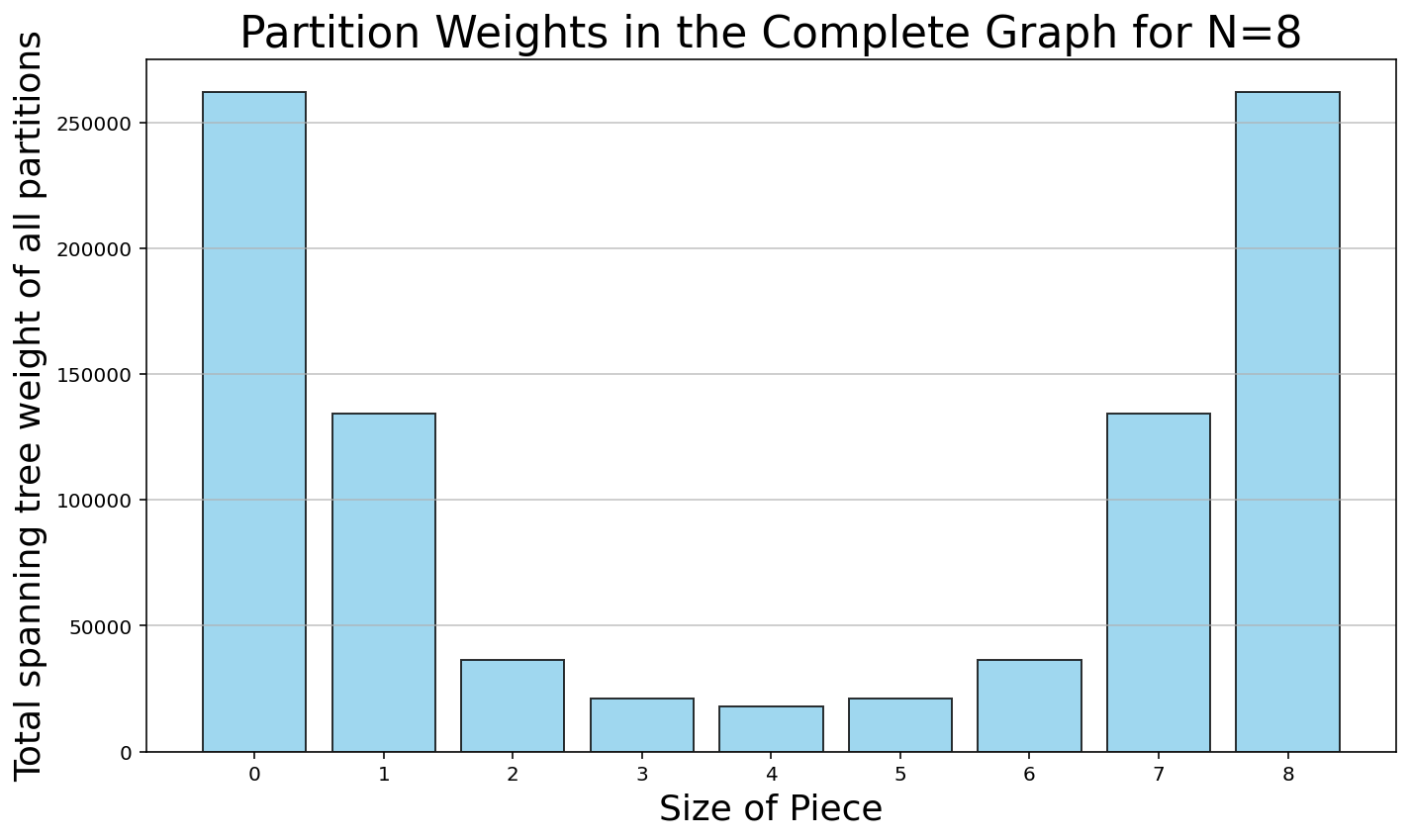}\hspace{.3in} 
    \includegraphics[width=0.45\linewidth]{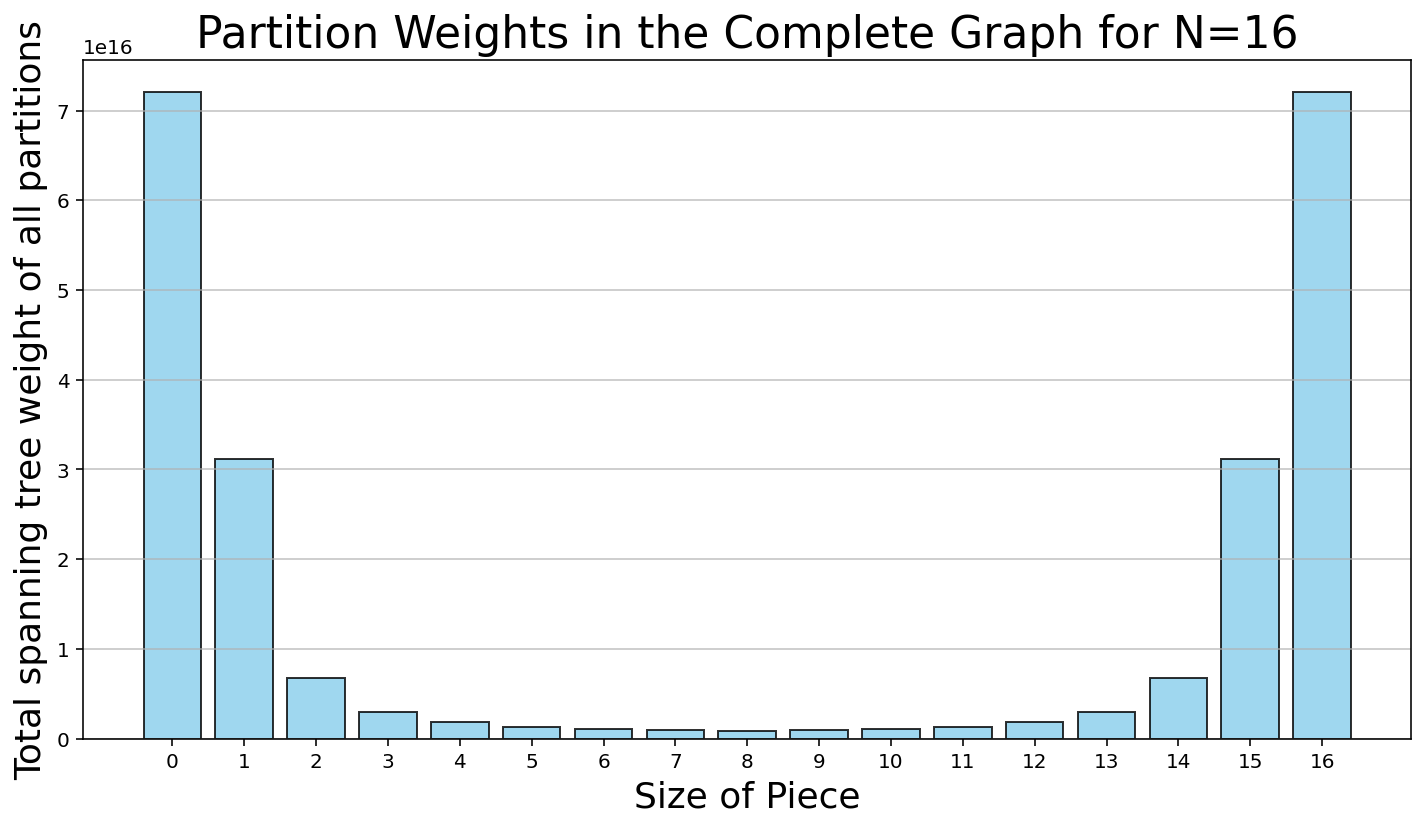}
    \caption{Total spanning tree weights of all partitions of size $i,N-i$ on the complete graph $K_N$ with $N=8$ (left) and $N=16$ (right).}
    \label{fig:weights}
\end{figure}

As Section~\ref{sec:split} demonstrates, even the smallest amount of slack can cause balanced partitions to be exponentially suppressed at stationarity.  This potentially poses a problem for any algorithms sampling spanning trees until a splittable one is found.  In this section we focus on Algorithms~\ref{alg:simple} and~\ref{alg:updown}, which are efficient only when balanced partitions (or balanced forests) are not too unlikely compared to the total weight where balance is not required.

\subsection{Splittability of the complete graph}

We now show that in the complete graph a polynomial fraction of spanning trees are splittable. 
This is equivalent to showing that a polynomial fraction of partitions of the complete graph are balanced (Lemma~\ref{lem:equiv}), which implies that we can find a balanced partition of the complete graph in polynomial time by generating random spanning trees and using rejection sampling.

Balanced forests tend not to be the mode of a histogram showing the probabilities of splits of various sizes. On the complete graph, balanced splits contain the fewest forests out of all possible partition sizes, as shown in Figure~\ref{fig:weights}. This means that standard techniques showing that this distribution is log-concave or unimodal will certainly fail and the calculations need to be more precise.

Instead of such combinatorial techniques, which are less likely to apply here, we use precise calculations to determine the probability of splittability on the complete graph. For random graphs, we  use conditional probability on the number of edges to analyze settings similar to the complete graph, create precise calculations there, and use our results about the complete graph to draw conclusions about the probability of splittability in more general random graphs.
 
We first  prove our result on the complete graph  because we will use it in the proof of the stronger (more general) result on dense random graphs that follows.
Recall Theorem~\ref{better_complete}, stated in the introduction:

\begin{theorem}
\label{thm:complete}
\label{better_complete}
    Let $N = nk$. 
    The probability that a random spanning tree of $K_N$ is $k$-splittable is 
    $(1+o(1))C_k n^{-(k-1)/2}$,
    where  $C_k$ depends only on $k$.
\end{theorem}

\begin{proof}
By Cayley's Theorem \cite{Cayley} the number of spanning trees in $K_N$ is $N^{N-2}$. Each splittable spanning tree can be split uniquely into balanced pieces (Lemma~\ref{uniqueness_lemma}); therefore, each one is uniquely defined by the partition it creates, the spanning trees on the $K_{n}$s in those partitions, and the $k-1$ edges between different parts. 
The number of unlabeled $k$-partitions of $N$ vertices is $$\frac{\binom{N}{n,n,\dots,n}}{k!} \ = \ \frac{(kn)!}{(n!)^k k!}.$$ 

The edges between the pieces in the partition can be considered a spanning tree on a separate graph whose vertices are the $k$ pieces of the partition. By another application of Cayley's Theorem, there are $k^{k-2}$ such spanning trees. There are $n^2$ ways to choose each of the $k-1$ edges, because each edge connects two components of size $n$. So, given the partition, there are $k^{k-2}(n^2)^{k-1}$ ways to choose the edges between the pieces. 
Again by Cayley's Theorem, the number of ways to choose spanning trees of $k$ copies of $K_n$ is $(n^{n-2})^k$.
Putting these together, the number of $k$-splittable spanning trees $\tau_k(K_N)$ of $K_N$ is 
$$\tau_k(K_N) \ = \ (n^{n-2})^k \frac{(n^2)^{k-1}k^{k-2}}{k!}\frac{(kn)!}{(n!)^k}.
$$
Robbins's version of Stirling's formula~\cite{Robbins1955} asserts that
$$
x! = \sqrt{2\pi} x^{x+1/2}e^{-x}e^{r_x},
$$
where $1/(12x+1) \leq r_x \leq 1/12x$.
Dividing by the number of spanning trees in $K_N$ and applying Robbins's formula to each factorial, the probability a tree is splittable is


\begin{align*}
     \frac{\tau_k(K_{kn})}{{(kn)}^{kn-2}} 
    \ &= \ \Big(e^{-r_n} \frac{1}{\sqrt{2\pi n}}\left(\frac{e}{n}\right)^n\Big)^k~
    e^{-r_k} \frac{1}{\sqrt{2\pi k}} \left(\frac{e}{k}\right)^k \\
   & \hspace{.05in} \times 
    e^{r_{kn}} \sqrt{2\pi kn} \left(\frac{kn}{e}\right)^{kn}
    \left( \frac{(n^{n-2})^k(n^2)^{k-1}k^{k-2}}{{(kn)}^{kn-2}} \right)    \\
    &= \ e^{r_{kn}-kr_n - r_k}~n^{-\frac {k-1}2}\bigg(\frac{e}{\sqrt{2\pi}}\bigg)^k \\
    &= (1+o(1))e^{-r_k} n^{-\frac {k-1}2}\bigg(\frac{e}{\sqrt{2\pi}}\bigg)^k 
\end{align*}

\noindent 
Setting $C_k = e^{-r_k}\Big(\frac{e}{\sqrt{2\pi}}\Big)^k$, the theorem follows.
\end{proof}

As long as $n \geq k$, the $1+o(1)$ in the theorem lies between $e^{-1/12}$ and $1$. 
We note that our bound for the probability that a random spanning tree is $k$-splittable on the complete graph  is $\Theta(1/n^{(k-1)/2})$, better than the best known bound of $\Omega(n^{2-2k})$ on the grid. 
However, Algorithm~\ref{alg:simple}, which generates a $k$-partition from a spanning tree, uses rejection sampling twice: once when determining if a spanning tree is $k$-splittable, and once to account for the fact that each $k$-balanced partition may be generated by many different spanning trees. This second step of rejection sampling is maximally inefficient in the complete graph, 
with an acceptance probability of $1/(k^{k-2}n^{2k-2})$.


\subsection{Dense random graphs}
We now show that in a dense random graph a polynomial fraction of spanning trees are splittable.  
This is equivalent to showing that a polynomial fraction of partitions of a dense random graph are balanced (Lemma~\ref{lem:equiv}), which implies that we can find a balanced partition of a dense random graph in polynomial time by generating random spanning trees and using rejection sampling.

While each spanning tree (with the vertices labeled) is equally likely to appear in a random graph, it is not true that any two spanning trees are equally likely to be selected as a random spanning tree of a random graph. The probability that any given spanning tree is selected as a random spanning tree is proportional to $\E[1/\tau(G)\,|\,T\in G]$, which varies for different spanning trees. 
Spanning trees in graphs that contain many other spanning trees are less likely to be selected. 

For example, 
if we choose a random graph on four vertices, then choose a random spanning tree of that graph, any labeled star with three leaves is more likely to be selected than any labeled path with three edges. 
\eat{We can see that adding $1$ edge to a path of $4$ vertices can create a graph with $3$ or $4$ spanning trees, while adding $1$ to a star graph always creates a graph with $3$ spanning trees. Since adding $2$ or more edges always creates $K_4$ or $K_4$ minus one edge, we see that}
This is because random graphs containing a star have fewer spanning trees than those containing a given path.
\eat{, so stars appear more often in random spanning tres of random graphs on $4$ vertices than they do in random labeled trees on $4$ vertices.} 

Our main theorem is as follows: 

\begin{theorem}
\label{thm:dense}
    Let $N = nk$ and let $G$ be a random connected graph in either the $G(N,p)$ model or the $G(N,m)$ model with $m=p\binom{N}{2}, p = \Omega(1 / \log n)$. Then the probability that a random spanning tree of $G$ is $k$-splittable is at least $C_k n^{-k/2 - O(1)}$,
    where $C_k$ depends only on $k$.
\end{theorem}

Theorem~\ref{thm:dense} 
will follow from the next two theorems. We first prove that a random spanning tree in a connected graph in the $G(N,m)$ model is splittable with probability at least 
$n^{-\frac{k}{2}-O(1)}$ when $m = \Omega( \binom{N}{2} / \log n )$ in Theorem \ref{mn}. 
We then use that result to prove the corresponding result for the $G(N, p)$ model
when $p = \Omega( 1 / \log n )$ in Theorem~\ref{np}.
In the statements of both theorems we fix $p = \Omega(1 / \log n)$,
and in the first theorem we set $m = p \binom{N}{2}$.
We note that in the $p = \Theta(1 / \log n)$ regime, the $O(1)$ in the exponent will depend on $p$, 
whereas in the $p = \Omega(1)$ regime, the $O(1)$ in the exponent will disappear,  
and~$C_k$ will depend on $p$ as well as $k$.

\begin{theorem} 
Let $N = nk$ and let $G$ be a random connected graph of $m$ edges on $N$ vertices. 
Further, let $p = m / \binom{N}{2}$,
and $T_0$ be a random spanning tree of $G$.
If $m = \Omega(\binom{N}{2} / \log n)$, 
then the probability that $T_0$ is $k$-splittable is at least 
$C_k n^{-\frac{k}{2}-O(1)}$,
where $C_k$ depends on $k$.
More precisely, the probability is at least $C_k n^{(1-k)/2}\exp({-(2+o(1))/p})$, an increasing function of $p$.
\label{mn}
\end{theorem}

\begin{proof}
Let $\tau_k(G)$ be the number of spanning trees of $G$ that are $k$-splittable. 
Let $\tau_{\max}(m, N)$ be the maximum number of spanning trees in a graph of $m$ edges on $N$ vertices.
For a fixed $G$ the probability that $T_0$ is $k$-splittable is 
$${\tau_k(G)} / {\tau(G)} \geq {\tau_k(G)} / {\tau_{\max}(m, N)}.$$
For a random connected graph $G$, the probability that a random spanning tree is $k$-splittable is the expected value of this ratio,
\begin{align*}
    \E\left[\frac{\tau_k(G)}{\tau(G)}\,|\,G\text{ connected}\right] \ \geq \ \frac{1}{\tau_{\max}(m, N)} \E[\tau_k(G)\,|\,G\text{ connected}] 
\end{align*}
We can express $\E[\tau_k(G)\,|\,G\text{ connected}]$ as
$\sum_{T ~k\mathrm{-splittable}} \P[T \in G\,|\,G\text{ connected}]$.
Using Theorem~\ref{better_complete} and the fact that
$\P[T \in G\,|\,G\text{ connected}] \geq \P[T \in G],$
\begin{align*}
\E\left[\frac{\tau_k(G)}{\tau(G)}\,|\,G\text{ connected}\right] & \ \geq \ \frac{1}{\tau_{\max}(m, N)} \sum_{T ~k\mathrm{-splittable}} \P[T \in G\,|\,G\text{ connected}] \\
&\ \geq \ \frac{1}{\tau_{\max}(m, N)} \sum_{T ~k\mathrm{-splittable}} \P[T \in G] \\
        & \ = \ \frac{C_k}{\tau_{\max}(m, N)} N^{N-2} n^{\frac{1-k}{2}} \P[T \in G],
\end{align*}

\noindent Grimmett \cite{grimmett_upper_1976} has shown that 
$$\tau_{\max}(m, N) \leq \frac{1}{N}\left(\frac{2m}{N-1}\right)^{N-1} = p^{N-1}N^{N-2}.$$
Therefore, the probability that our tree $T_0$ is $k$-splittable is at least 
\begin{align*}
    \frac{C_k}{p^{N-1}N^{N-2}} N^{N-2} n^{(1-k)/2} \P[T \in G]
    \ = \ \frac{C_k}{p^{N-1}}~n^{(1-k)/2}~ \P[T \in G],   
\end{align*}
after absorbing a factor of $k$ into $C_k$.
It will now suffice to show that $\P[T \in G] \geq p^{N-1} \exp({-(2+o(1))/p})$, because $\exp({-(2+o(1))/p}) = n^{-O(1)}$. The probability that any set $\{e_1,\dots,e_{N-1}\}$ of $N-1$ edges is in $G$ is 
\begin{align*}
    \prod_{i=1}^{N-1} \P[e_i\in G\,|\,e_j\in G\,\,\forall j<i] \ &=\ 
    \prod_{i=0}^{N-2}\frac{m-i}{\binom{N}{2}-i} \\
    &\ \geq \ \left(\frac{m-N+2}{\binom{N}{2}}\right)^{N-1} \\
    &\ =\  \left(p - \frac{2(N-2)}{N(N-1)}\right)^{N-1} \\
    &\ \geq \ p^{N-1} \left(1 - \frac{2}{p(N-1)}\right)^{N-1} \\
    & \ = \ p^{N-1} \exp({-(2+o(1))/p}),
\end{align*}
which proves the theorem.
\end{proof}
To prove the corresponding theorem for $G(N, p)$, we will use the following lemma. \vspace{-.1in}
\begin{lemma}
    Let $G$ be a random graph in $G(N,p)$ conditioned on being connected.
    If $p \leq 1/2$, then with probability at least $1/3$, $G$ has at least $\lfloor p\binom{N}{2}\rfloor$ edges.
\end{lemma}
\begin{proof}
    Let $C$ be the event that $G \in G(N,p)$ is connected and $L$ be the event that $G$ has at least $\lfloor p\binom{N}{2}\rfloor$ edges. 
    We want to prove that $\P(L|C) = \P(L)\ \P(C|L)\ /\ \P(C) \geq 1/3$. 
    Because $\P(C|L) \ / \ \P(C) \geq 1$ by Harris's Inequality \cite{harris}, 
    it will suffice to prove that $\P(L) \geq 1/3$ for an
    arbitrary random graph $G$ in the $G(N,p)$ model.     
    Letting $M \ = \  \binom{N}{2}$ and $m = \lfloor pM \rfloor$, we can show that for each $j \leq m$, 
    $\P[|E(G)| = m+j] \geq \P[|E(G)| = m-j]/2$. 
    We see
    \begin{align*}
        \frac{\P[|E(G)| = m+j]}{\P[|E(G)| = m-j]}
         &\ \geq \  \frac{\binom{M}{m+j}(\frac{m}{M})^{m+j}(1-\frac{m}{M})^{M-m-j}}
                 {\binom{M}{m-j}(\frac{m}{M})^{m-j}(1-\frac{m}{M})^{M-m+j}} \\
         &\ = \ \frac{(M-m+j)\cdots(M-m-j+1)}{(M-m)^{2j}}
            \frac{m^{2j}}{(m+j)\cdots(m-j+1)} \\
         & \ =\  \frac{1 + \frac{j}{M-m}}{1 + \frac{j}{m}}
         \prod^{j-1}_{i=1} \frac{1 - (\frac{i}{M-m})^2} {1 - (\frac{i}{m})^2} \\
         & \ \geq \ \frac{1 + \frac{j}{M-m}}{1 + \frac{j}{m}},
    \end{align*}
    where the last inequality follows from the fact that $p \leq 1/2$ and, therefore, $M-m \geq m$. 
    The final expression on the right-hand side is at least $1/2$, proving the result.
\end{proof}

\begin{theorem}
    Let $N = nk$. 
    Let $G$ be a random graph in the $G(N,p)$ model conditioned on being connected. 
    Let~$T$ be a random spanning tree of $G$.
    If $p = \Omega( 1 / \log n )$, then the probability that $T$ is $k$-splittable is at least 
    $C_k n^{(1-k)/2}\exp({-(2+o(1))/p})$  
    $= C_k n^{-k/2 - O(1)}$,
    where $C_k$ depends only on $k$. 
\label{np}
\end{theorem}

\begin{proof}
    Let $p' = \min(p, 1/2)$. With probability at least $1/3$, $G$ has $\geq p'\binom{N}{2}-1$ edges (using monotonicity if $p > 1/2$).  

    For any fixed $m \geq p'\binom{N}{2}-1$, if we condition on $G$ having $m$ edges, then both $G(N, m)$ and $G(N, p)$ have the uniform distribution over connected graphs. 
    We now apply Theorem~\ref{mn} for each such $m$. Recalling that the bound in Theorem~\ref{mn} is increasing in $m/\binom{N}{2}$, the probability that $T$ is $k$-splittable is at least
    \begin{align*}
        C_k n^{(1-k)/2} \exp({-(2+o(1))/p}) = C_k n^{-k/2 - O(1)},
    \end{align*}
    where we have absorbed factors of $1/3$ and $\exp({-2+o(1)})$ (the minimum over $p$ of $\exp({(2+o(1))/p - (2+o(1))/p'})$) into $C_k$. 
    The theorem follows.
\end{proof}

\noindent Note that Theorem~\ref{thm:dense} 
follows immediately from Theorems~$\ref{mn}$ and~$\ref{np}$.

\section{Random Spanning Tree Diameters}
One might ask how these graph classes compare in practice, so we include a table of 
values showing our simulations of the average diameter of random spanning trees on random and complete graphs, and minimum spanning trees of complete graphs with random values. These empirical results helped motivate our questions surrounding random structures on random graphs.

Each entry in Table \ref{table:empirical} is the mean diameter of random spanning trees, taken over $3000$ trials. In the minimum spanning tree (MST) case, a complete graph on $N$ vertices is given random weights uniformly in~$[0,1]$, and the diameter of the minimum spanning tree is taken, and in the other cases, a $G(N,p)$ random graph is created each trial, and the diameter of a uniform random spanning tree is measured. Trials where the random graph was not connected are not included in the average.

We see that the smallest diameters occur on minimum spanning trees, and random spanning trees on $G(N,p)$ graphs with smaller values of $p$ have clearly smaller diameters as $p$ decreases. The expected diameter of a random minimum spanning tree scales with $\Theta(N^{1/3})$\cite{SOURCE-MST} while the expected diameter of a random spanning tree of $K_N$  scales with $\Theta(N^{1/2})$\cite{SOURCE-UST}, and the scaling limit is unknown for a random spanning tree of a random graph. We note that spanning trees in $G(N,1/2)$ graphs have close enough average diameter to $K_N$ that a difference is not observed in this data.


\section{Conclusions and Open Problems}

We have shown that random spanning trees are polynomially likely to be splittable, in both the complete graph and  dense random graphs. These results allow polynomial-time sampling techniques to guarantee an inverse polynomial probability of sampling a valid balanced partition. 

It 
is unknown whether the ratio $\tau(G)/\tau_k(G)$ is bounded by some polynomial for many classes of random graphs. 
For small enough $p$, specifically $p=O(N^{-2})$, the probability that a random tree from a random graph conditioned on connectivity is $k$-splittable is inverse polynomial for a trivial reason --- the expected number of edges outside of a conditioned tree is constant, so the number of spanning trees is bounded by a constant in $G(N,m)$, and bounded by a constant in sufficiently high probability in $G(N,p)$; therefore, for every tree $T$, the expectation $\E[1/\tau(G)\,|\,T\in G]$ is bounded below by some inverse polynomial. We conjecture that polynomial probability of splittability will hold for all $p$ above the connectivity threshold.
We also observe a link between random spanning trees of random graphs and minimum spanning trees with random weights in the complete graph. We conjecture that random spanning trees of random graphs will have smaller diameter than random spanning trees of the complete graph, as this occurs on random minimum spanning trees  \cite{SOURCE-MST,SOURCE-UST} and in our empirical results on random graphs.

\section*{Acknowledgements}
Jacob Platnick and Dana Randall were supported in part by NSF award CCF-2106687 and the US Army Research Office Multidisciplinary University Research Initiative award W911NF-19-1-0233.
\begin{table} 
\begin{center}
\resizebox{6.in}{!}{
$\begin{array}{c||S|S|S|S|S|S}
  N    & {MST} &{G(N,{\log(N)}/{N})}  &  {G(N,{2\log(N)}/{N})}  &  {G(N,{1}/{10})}  & {G(N,{1}/{2})}  & {K_N} \\
  \hline
  20 &9.80\ \ \ \ \ \ \  \   &\ \ \ \ \ \ \ \ 10.25 \ \ \ \ \ \ \ \ &\ \ \ \ \ \ \ \  10.37\ \ \ \ \ \ \  \  &\ \ \ \ \ \ \ \   10.10  \ \ \ \ \ \ \ \ &\ \ \ \ \ \ \ \ 10.43\ \ \ \ \ \ \  \ &\ \ \ \ \ \ \  10.49  \\
  40 &14.75  & 15.86   & 16.32   & 16.07    & 16.42  &16.34 \\
  80 &21.26  & 24.51   &24.76    & 24.74    &  25.01    & 24.91  \\
\end{array}$ 
}
\end{center}
\caption{Empirical average spanning tree diameters on different randomization models}
\label{table:empirical}
\end{table}
\bibliographystyle{splncs04}
\bibliography{balanced.bib}

\end{document}